\documentclass[sigconf]{acmart}
\AtBeginDocument{%
  \providecommand\BibTeX{{%
    \normalfont B\kern-0.5em{\scshape i\kern-0.25em b}\kern-0.8em\TeX}}}

\copyrightyear{2022}
\acmYear{2022}
\setcopyright{acmcopyright}
\acmConference[ISSAC '22] {Proceedings of the 2022 International Symposium on Symbolic and Algebraic Computation}{July 4--7, 2022}{Villeneuve-d'Ascq, France.}
\acmBooktitle{Proceedings of the 2022 Int'l Symposium on Symbolic and Algebraic Computation (ISSAC '22), July 4--7, 2022, Villeneuve-d'Ascq, France}
\acmPrice{15.00}
\acmISBN{978-1-4503-8688-3/22/07}
\acmDOI{10.1145/3476446.3535492}

\theoremstyle{definition}
\newtheorem{notation}{Notation}
\newtheorem{remark}{Remark}

\usepackage{bm}
\usepackage[ruled,vlined,linesnumbered]{algorithm2e}
\usepackage{enumitem}

\DeclareMathOperator{\ord}{ord}

\DeclareMathOperator{\Znn}{\mathbb{Z}_{\geqslant 0}}

\begin{document}
\fancyhead{}
\title[On realizing differential-algebraic equations by rational dynamical systems]{On realizing differential-algebraic equations\\ by rational dynamical systems}

\author{Dmitrii Pavlov}
\email{dmmpav@gmail.com}
\affiliation{%
  \institution{Faculty of Mechanics and Mathematics,\\ Moscow State University}
  \city{Moscow}
  \country{Russia}
  \postcode{119991}
}

\author{Gleb Pogudin}
\email{gleb.pogudin@polytechnique.edu}
\affiliation{%
  \institution{LIX, CNRS, \'Ecole Polytechnique,\\ Institute Polytechnique de Paris}
  \city{Palaiseau}
  \country{France}
  \postcode{91120}
}


\begin{abstract}
  Real-world phenomena can often be conveniently described by dynamical systems (that is, ODE systems in the state-space form). However, if one observes the state of the system only partially, the observed quantities (outputs) and the inputs of the system can typically be related by more complicated differential-algebraic equations (DAEs). Therefore, a natural question (referred to as the realizability problem) is: given a differential-algebraic equation (say, fitted from data), does it come from a partially observed dynamical system? A special case in which the functions involved in the dynamical system are rational is of particular interest. For a single differential-algebraic equation in a single output variable, Forsman has shown that it is realizable by a rational dynamical system if and only if the corresponding hypersurface is unirational, and he turned this into an algorithm in the first-order case.
  
   In this paper, we study a more general case of single-input-single-output equations. We show that if a realization by a rational dynamical system exists, the system can be taken to have the dimension equal to the order of the DAE. We provide a complete algorithm for first-order DAEs. We also show that the same approach can be used for higher-order DAEs using several examples from the literature.
\end{abstract}

\ccsdesc[500]{Computing methodologies~Symbolic calculus algorithms}

\keywords{differential-algebraic equations, rational dynamical systems, realization theory}

\maketitle

\section{Introduction}

Many processes in the sciences and engineering are described by systems of differential equations.
One of the prominent classes of systems of differential equations are systems in~\emph{the state-space form}:
\begin{equation}\label{eq:stspace1}
  \mathbf{x}' = \mathbf{f}(\mathbf{x}, \mathbf{u}),
\end{equation}
where $\mathbf{x} = (x_1, \ldots, x_n)$ are the unknowns describing the state of the system (\emph{state variables}), $\mathbf{u} = (u_1, \ldots, u_m)$ are the unknowns representing external forces (\emph{input variables}), and $\mathbf{f} = (f_1, \ldots, f_n)$ are the functions describing how the rate of change of the state depends on the state and external inputs.

A typical experimental setup contains an assumption that the functions $\mathbf{u}$ are known while the states $\mathbf{x}$ may be only partially observed.
In order to encode this constraint into the system, one augments~\eqref{eq:stspace1} with the \emph{output variables} $\mathbf{y} = (y_1, \ldots, y_\ell)$ and the equations describing the observations
\begin{equation}\label{eq:stspace2}
  \mathbf{y} = \mathbf{g}(\mathbf{x}, \mathbf{u}).
\end{equation}
Thus, one typically has time course data for $\mathbf{y}$ and $\mathbf{u}$ only, not for $\mathbf{x}$.
Therefore, one may be able to fit the equations satisfied by $\mathbf{y}$ and $\mathbf{u}$, but not the original~\eqref{eq:stspace1} and~\eqref{eq:stspace2}.

The question of reconstructing a system in the state-space form (that is, \eqref{eq:stspace1} with~\eqref{eq:stspace2}) which explains a given set of relations between $\mathbf{y}$ and $\mathbf{u}$ is called~\emph{the realizability problem} and it is one of important problems in control theory.
This problem is well studied for linear systems, see e.g.~\cite{Kalman1963, Silverman1971}.
In the nonlinear case, there are several versions of the problem depending on where $\mathbf{f}$ and $\mathbf{g}$ are sought.
Two popular classes considered in this paper are rational functions and input-affine rational functions as in~\cite{SontagYuan, NS1, NS2}, but one could also consider algebraic, analytic, or smooth functions~\cite{Algebraic1, Kotta2018, Sussmann1976, Schaft1986}.
From the constructive standpoint, the case of single-output-no-input systems (for which rational and input-affine rational functions coincide) has been considered by Forsman~\cite{Forsman1993}.
He has shown that a DAE in $y$ can be realized by a rational system in the state-space form if and only if the corresponding hypersurface is unrational.
In particular, an algorithm for the first-order DAE was proposed.
For higher order, although the general problem of assessing unirationality is notoriously hard, many theoretical results are available~\cite{rational_book} which could be used to find sufficient or necessary conditions for realizability.

The goal of the present paper is to consider the realization problem in the presence of inputs.
Our contribution is two-fold.
On the theoretical side, we prove that if a DAE of order $h$ can be realized by a system in the state-space form, then it can be realized by a system of dimension $h$ (that is, by a locally observable one).
This result is related to a theorem by Sussmann~\cite{Sussmann1976} and its analogues for rational realizations~\cite{NS2} (see also~\cite{minimal, Kotta2018}) which state that, for a realization problem (analytic or rational), if a realization exists, it can always be taken to be observable at the expense of allowing a realization to be defined not on an affine space but on an arbitrary variety.
We achieve only local observability but guarantee the existence of a realization defined on an affine space.
Note that our result is sharp in the sense that there exist realizable DAEs without observable realizations by a system of the state-space form, see~\cite[Section~4]{Forsman1993}.

On the computational side, we use the developed theory to propose algorithms for solving both rational and input-affine rational realization problems for first-order single-output-single-input DAEs.
We also show, using examples from the literature, that an approach similar to the one we use for the first-order case, can be successfully applied for DAEs of higher order as well (see Section~\ref{sec:examples} and Appendix).

The rest of the paper is organized as follows. 
Section~\ref{sec:prelim} contains a precise statement of the realizability problem.
Theoretical results are stated and proved in Section~\ref{sec:theory}.
Sections~\ref{sec:zero-ord} and~\ref{sec:first-ord} contain our algorithms and proofs of their correctness.
Finally, Section~\ref{sec:examples} and Appendix contain several worked out examples from the literature.
Maple worksheets with the examples are available at~\cite{MapleNotebooks}.


\section{Preliminaries}\label{sec:prelim}

In this paper, we will use the language of differential algebra which we introduce in Section~\ref{sec:diffalg}.
The main problem studied in this paper, the realization problem, can be viewed as an inverse problem to the differential elimination problem for dynamical systems, so we first introduce the elimination problem in Section~\ref{sec:elim}, and then define the realization problem in Section~\ref{sec:realize}.

Throughout Sections~\ref{sec:prelim} and~\ref{sec:theory}, $k$ is an algebraically closed field of zero characteristic (e.g., $\mathbb{C}$).
In Sections~\ref{sec:zero-ord} and~\ref{sec:first-ord}, it will be additionally assumed to be constructive.
For affine varieties $X$ and $Y$, a rational map $\varphi$ from $X$ to $Y$ will be denoted by $\varphi\colon X \dashrightarrow Y$.
The corresponding map $k(Y) \to k(X)$ will be denoted by $\varphi^*$.


\subsection{Differential algebra}\label{sec:diffalg}

\begin{definition}[Differential rings and fields]\label{def:diffrings}
  A {\em differential ring} $(R,\,')$ is a commutative ring with a derivation $'\!\!:R\to R$, that is, a map such that, for all $a,b\in R$, $(a+b)'=a'+b'$ and $(ab)'=a'b+ab'$. 
  A {\em differential field} is a differential ring that is a field.
  For  $i>0$,  $a^{(i)}$ denotes the $i$-th order derivative of $a \in R$.
  An element $a \in R$ of a differential ring is said to be \emph{a constant} if $a' = 0$.
\end{definition}

\begin{notation}
  Let $x$ be an element of a differential ring and $h \in \Znn$. We introduce
  \[
      x^{(<h)} := (x, x', \ldots, x^{(h - 1)})\;\text{ and }\;
      x^{(\infty)} := (x, x', x'', \ldots).
  \]
  $x^{(\leqslant h)}$ is defined analogously.
\end{notation}

\begin{definition}[Differential polynomials]
  Let $R$ be a differential ring. 
  Consider a ring of polynomials in infinitely many variables
  \[
  R[x^{(\infty)}] := R[x, x', x'', x^{(3)}, \ldots]
  \]
  and extend the derivation from $R$ to this ring by $(x^{(j)})' := x^{(j + 1)}$.
  The resulting differential ring is called \emph{the ring of differential polynomials in $x$ over $R$}.
  The ring of differential polynomials in several variables is defined by iterating this construction.
\end{definition}

\begin{notation}
  For a differential polynomial $p \in k[x^{(\infty)}]$, we define the order of $p$ with respect to $x$ (denoted by $\ord_x p$) as the largest integer $i$ such that $x^{(i)}$ appears in $p$.
  If no such $i$ exists, we define $\ord_x p = -1$.
\end{notation}

\begin{definition}[Differential ideals]
   Let $R$ be a differential ring.
   An ideal $I \subset R$ is called a \emph{differential ideal} if $a' \in I$ for every $a \in I$.
  
   One can verify that, for every $f_1, \ldots, f_s \in R$, the ideal
   \[
      \langle f_1^{(\infty)}, \ldots, f_s^{(\infty)} \rangle
    \]
    is a differential ideal.
    Moreover, this is the minimal differential ideal containing $f_1, \ldots, f_s$, and we will denote it by $\langle f_1, \ldots, f_s \rangle^{(\infty)}$.
\end{definition}

\begin{notation}[Saturation]
  Let $R$ be a ring, $I \subset R$ be an ideal, and $a \in R$.
  We introduce
  \[
      I \colon a^{\infty} := \{b \in R \mid \exists N \in \mathbb{Z}_{\geqslant 0}  \colon a^Nb \in I\},
  \]
  which is also an ideal in $R$.
\end{notation}


\subsection{Direct problem: Differential elimination}\label{sec:elim}

Consider an ODE system in the so-called \emph{state-space form}:
 \begin{equation}\label{eq:ODEmodel}
   \Sigma = \begin{cases}
    \mathbf{x}' = \mathbf{f}(\mathbf{x}, \mathbf{u}),\\
    y = g(\mathbf{x}, \mathbf{u}),
    \end{cases}
   \end{equation}
where 
\begin{itemize}
    \item $\mathbf{x} = (x_1, \ldots, x_n)$ and $\mathbf{u} = (u_1, \ldots, u_m)$ are the vectors of state and input variables, respectively; the number~$n$ is called the \emph{dimension} of the system;
    \item $y$ is a single output variable (there may be several outputs but in this paper we restrict ourselves to the single-output case);
    \item $\mathbf{f} = (f_1, \ldots, f_n)$, where $f_1, \ldots, f_n \in k(\mathbf{x}, \mathbf{u})$, $g \in k(\mathbf{x}, \mathbf{u})$ and $k$ is a constant coefficient field.
\end{itemize}

The system~\eqref{eq:ODEmodel} is called \emph{input-affine} if $f_1, \ldots, f_n, g$ are affine (i.e. polynomials of degree~$1$) with respect to $\mathbf{u}$.

Bringing $f_1, \ldots, f_n, g$ to the common denominator,  write $\mathbf{f} = \mathbf{F}/Q$ and $g = G/Q$, where $\mathbf{F} = (F_1, \ldots, F_n)$ and $F_1, \ldots, F_n, G, Q \in k[\mathbf{x}, \mathbf{u}]$.
Consider the differential ideal
\begin{equation}\label{eq:Isigma_gens}
   I_\Sigma := \langle Qx_1' - F_1, \ldots, Qx_n' - F_n, Qy - G\rangle^{(\infty)} \colon Q^\infty \subset k[\mathbf{x}^{(\infty)}, y^{(\infty)}, \mathbf{u}^{(\infty)}]
\end{equation}
which is prime by~\cite[Lemma~3.2]{HOPY2020}.
Note that every element of $I_\Sigma$ vanishes on every analytic or formal power series solution (w.r.t~$\mathbf{x}, \mathbf{u}$ and~$y$) of~\eqref{eq:ODEmodel}.

\begin{definition}[Input-output equation]
We define the \emph{ideal of input-output relations of~\eqref{eq:ODEmodel}} as
\[
    J_{\Sigma} := I_{\Sigma} \cap k[y^{(\infty)}, \mathbf{u}^{(\infty)}].
\]
These relations play an important role in control theory~\cite{Conte2007, Sontag1998} since they only involve functions for which experimental data is typically available (i.e., inputs and outputs).
For the single-output case, which we consider in this paper, it is known~\cite[Remark 2.20]{structidjl} that $J_\Sigma$ is ``almost principal'', that is, if $P$ is an irreducible nonzero polynomial in $J_\Sigma$ of minimal possible order w.r.t~$y$ (which is unique up to a multiplicative constant), then
\begin{equation}\label{eq:ioeq_ideal}
  J_\Sigma = \langle P \rangle^{(\infty)} \colon H^\infty,
\end{equation}
where $H = \frac{\partial P}{\partial y^{(h)}}$ and $h = \ord_y P$.
Such $P$ typically is referred to as \emph{the input-output equation} of $\Sigma$ and, thanks to~\eqref{eq:ioeq_ideal}, fully characterizes the input-output behavior of~$\Sigma$.
\end{definition}

\begin{remark}[Multiple-input case]
  Throughout the paper, we will focus on the case of single input.
  We expect that the same methods will work for the multiple-input case, and we plan to elaborate on this in an extended version of the paper.
\end{remark}



\subsection{Inverse problem: Realization}\label{sec:realize}

Now we are ready to define \emph{the realization problem} which is the main problem of this paper.
\begin{description}
  \item[Input:] an irreducible differential polynomial $P \in k[y^{(\infty)}, \mathbf{u}^{(\infty)}]$, where $\mathbf{u} = (u_1, \ldots, u_m)$;
  \item[Output 1:]\label{out:rational} a system of the form~\eqref{eq:ODEmodel} such that $P$ is the input-output equation for this system or \texttt{NO} if there is no such system;
  \item[Output 2:]\label{out:affine} an input-affine system of the  form~\eqref{eq:ODEmodel} such that $P$ is the input-output equation for this system or \texttt{NO} if there is no such system;
\end{description}
We will refer to the cases of~\textbf{Output~1} and~\textbf{Output~2} as \emph{the rational realization problem} and~\emph{input-affine rational realization problem}, respectively.

If there is a realization of $P$, it is typically not unique (e.g., in can be composed with any invertible change of coordinates).
Therefore, after constructing a realization, one may want to perform change of variables in the resulting system to obtain a more ``interpretable'' or ``insightful'' realization (see~\cite{Chris, COB} for related results).
This second step if out of scope of the present paper.





\section{Theory}\label{sec:theory}

In the rest of the paper, we will use the notation $\mathbb{U} := \overline{k(u^{(\infty)})}$.

\subsection{General realizability criterion}

\begin{notation}[Lie derivatives]
\begin{itemize}
  \item[]
  \item Let $R \in K(u^{(\infty)})$ be a rational function in $u^{(\infty)}$ over a field $K$ (we will use $K = k(\mathbf{x})$).
  We define $D_u$ as
  \[
    D_u (R) := \sum\limits_{j = 0}^\infty u^{(j + 1)} \frac{\partial R}{\partial u^{(j)}}.
  \]

  \item Consider a system~$\Sigma$ as in~\eqref{eq:ODEmodel}.
  For any rational function $R \in k(\mathbf{x}, u^{(\infty)})$, we define the Lie derivative w.r.t.~$\Sigma$ by
  \[
    \mathcal{L}_{\Sigma} (R) := \sum\limits_{i = 1}^n f_i \frac{\partial R}{\partial x_i} + D_u (R).
  \]
  
  \item The Lie derivative of~$R$ of order~$i$ w.r.t. $\Sigma$ is obtained by iteratively applying the formula above and is denoted by~$\mathcal{L}_{\Sigma}^i (R)$.
 \end{itemize}
\end{notation}

\begin{remark}[Lie derivative as a derivation]\label{Liefacts}
  Consider
  the differential ideal $\widetilde{I}_\Sigma$ generated by $I_\Sigma$ in the differential ring $k(\mathbf{x}, u^{(\infty)})[(\mathbf{x}')^{(\infty)}, y^{(\infty)}]$ with derivation $'$.
  \begin{itemize}
      \item The primality of $I_\Sigma$ and \cite[Lemma~3.1]{HOPY2020} imply that $\widetilde{I}_\Sigma$ is proper and $\widetilde{I}_\Sigma \cap k[\mathbf{x}^{(\infty)}, y^{(\infty)}, u^{(\infty)}] = I_\Sigma$.
      \item For every $R \in k(\mathbf{x}, u^{(\infty)})$, we have $\mathcal{L}_{\Sigma}(R) - R' \in \widetilde{I}_\Sigma$.
      \item $k(x, u^{(\infty)})$ is a differential field w.r.t. the derivation~$\mathcal{L}_\Sigma$.
  \end{itemize}
\end{remark}

\begin{notation}[Corresponding hypersurface]
  Consider an irreducible differential polynomial $P(y, u) \in k[y^{(\infty)}, u^{(\infty)}]$ with $\ord_y P = h$.
  Then the hypersurface defined by $P = 0$ in the affine space with the coordinates $y, y',\ldots, y^{(h)}$ over the field $\mathbb{U}$ will be denoted by $\mathcal{H}_P$ and referred to as \emph{the corresponding hypersurface}.
\end{notation}

\begin{lemma}[Realizability criterion]\label{lem:linear_system}
  Let $P(y, u) \in k[y^{(\infty)}, u^{(\infty)}]$ be irreducible with $\ord_y P = h$.
  There exists a rational (resp. input-affine rational) realization of $P$ of dimension~$n$ if and only if there exists an integer $n$ and a dominant (i.e. such that its image is Zariski dense) map $\bm{\gamma} \colon \mathbb{A}_{\mathbb{U}}^{n} \dashrightarrow \mathcal{H}_P$ defined over~${k(u^{(\infty)})}$ with $\bm{\gamma} = (\gamma_0, \ldots, \gamma_h)$ and the coordinates in $\mathbb{A}_{\mathbb{U}}^n$ being $\mathbf{x} = (x_1, \ldots, x_n)$ such that $\gamma_0 \in k(\mathbf{x}, u)$ (resp., $\gamma_0 \in k(\mathbf{x}) + k(\mathbf{x})u$) and the following linear system in $Z_1, \ldots, Z_n$
    \begin{equation}\label{eq:linear_system_main}
  \begin{pmatrix}
  \gamma_1 - D_u(\gamma_0)\\
  \vdots\\
  \gamma_h - D_u(\gamma_{h - 1})
  \end{pmatrix}
  =
  \begin{pmatrix}
    \frac{\partial \gamma_{0}}{\partial x_1} & \ldots & \frac{\partial \gamma_{0}}{\partial x_n}\\
    \vdots & \ddots & \vdots\\
    \frac{\partial \gamma_{h - 1}}{\partial x_1} & \ldots & \frac{\partial \gamma_{h - 1}}{\partial x_n}
  \end{pmatrix}
  \begin{pmatrix}
    Z_1\\
    \vdots\\
    Z_n
  \end{pmatrix}
  \end{equation}
  has a solution in $k(\mathbf{x}, u)$ (resp., in $k(\mathbf{x}) + k(\mathbf{x})u$).
\end{lemma}

\begin{proof}
   Assume that $P$ is realizable by a system $\Sigma$ as in~\eqref{eq:ODEmodel} with the dimension of the state space being $n$.
   For every $i \geqslant 0$, we have $y^{(i)} - \mathcal{L}_\Sigma^i(g) \in \widetilde{I}_\Sigma$.
   Since $P \in I_\Sigma$, it is annihilated by~$g$ and its first~$h$ Lie derivatives w.r.t.~$\Sigma$. Thus, we have a map $\bm{\gamma}\colon\mathbb{A}_{\mathbb{U}}^n \dashrightarrow \mathcal{H}_P$ defined by 
   \[
      (x_1, \ldots, x_n) \mapsto (g, \mathcal{L}_{\Sigma}(g), \ldots, \mathcal{L}_{\Sigma}^h(g)).
  \]
  Since $P$ has the minimal order in $J_{\Sigma}$, the elements $g, \mathcal{L}_{\Sigma}(g), \ldots, \mathcal{L}_{\Sigma}^{h - 1}(g)$ are algebraically independent over $\mathbb{U}$ and, thus, $\bm{\gamma}$ is dominant.
  Finally, we observe that the vector $Z_i = f_i$ is a solution to~\eqref{eq:linear_system_main} by the definition of Lie derivative.
  
  In the other direction, assume that there exists such a dominant map $\bm{\gamma}$ and let $f_1, \ldots, f_n$ be a solution of~\eqref{eq:linear_system_main} in $k(\mathbf{x}, u)$ (resp., $k(\mathbf{x}) + k(\mathbf{x})u$).
  Consider a system $\Sigma$
  \[
  \begin{cases}
      x_1' = f_1(x_1, \ldots, x_n, u),\\
      \ldots\\
      x_n' = f_n(x_1, \ldots, x_n, u),\\
      y = \gamma_0(x_1, \ldots, x_n, u).
  \end{cases}.
  \]
  We claim that $P$ is the input-output equation for $\Sigma$.
  Indeed, since $f_1, \ldots, f_n$ is a solution of~\eqref{eq:linear_system_main}, we have $\mathcal{L}_{\Sigma}^i(\gamma_0) = \gamma_i$ for every $0 \leqslant i \leqslant h$.
  Therefore, $P \in J_{\Sigma}$.
  Since $\bm{\gamma}$ is dominant, $\gamma_0, \ldots, \gamma_{h - 1}$ are algebraically independent over $\mathbb{U}$, so $P$ is the irreducible element in $J_{\Sigma}$ of the lowest order, so it is the input-output equation.
\end{proof}


\subsection{Existence of a realization of minimal order}

The goal of this section is to prove the following theorem.

\begin{theorem}\label{thm:sussmann}
  Let $P(y, u) \in k[y^{(\infty)}, u^{(\infty)}]$ be an irreducible differential polynomial with $\ord_y P = h$.
  If there exists a rational (resp., input-affine rational) realization of $P(y, u)$, then there exists a rational (resp., input-affine rational) realization of $P(y, u)$ of dimension~$h$.
\end{theorem}

We start with the following lemma.

\begin{lemma}\label{lem:alg_dep}
  Let $\Sigma$ be a system of the form~\eqref{eq:ODEmodel} with the right-hand side being polynomial in $u$.
  Let $p_1, \ldots, p_s \in k(\mathbf{x})[u^{(\infty)}]$ be algebraically independent over $k(u^{(\infty)})$ such that, for every $1 \leqslant i \leqslant s$, $\mathcal{L}_{\Sigma}(p_i)$ is algebraic over $k(u^{(\infty)}, p_1, \ldots, p_s)$.
  
  Let $\mathcal{C}$ be the set of the coefficients of $p_1, \ldots, p_s$ considered as polynomials in $u^{(\infty)}$.
  Then $\operatorname{trdeg}_{k} k(\mathcal{C}) = s$.
\end{lemma}

\begin{proof}
  We will first prove the following statement: for every $p \in k(\mathbf{x})[u^{(\infty)}]$ algebraic over $F := k(u^{(\infty)}, p_1, \ldots, p_s)$, its coefficients as a polynomial in $u^{(\infty)}$ are also algebraic over $F$.
  We will prove this by induction on the number of monomials in $p$. 
  For a single monomial, the statement is true.
  Assume that there is more than one monomial. 
  Let $h = \ord_u p$. 
  By dividing by $u^{(h)}$ if necessary, we may assume that $u^{(h)} \not\mid p$, so $\frac{\partial p}{\partial u^{(h)}}$ has fewer monomials than $p$.
  Let $H := \max\limits_{1 \leqslant i \leqslant s}\ord_u p_i$.
  Let $h_0 := \max(1, H - h + 1)$ and $P := \mathcal{L}_{\Sigma}^{h_0} (p)$.
  The algebraic closure of $F$ is a differential field with respect to $\mathcal{L}_{\Sigma}$. 
  Thus, $P$, being the~$h_0$-th Lie derivative of~$p \in \overline{F}$, is also algebraic over $F$.
  We can write 
  \begin{equation}\label{eq:P_expansion}
    P = \frac{\partial p}{\partial u^{(h)}} u^{(h + h_0)} + Q, \quad \text{where}\quad Q \in k(\mathbf{x})[u^{(< h + h_0)}].
  \end{equation}
  Let $R \in k[u^{(\infty)}][X_1, \ldots, X_s, Y]$ be an irreducible polynomial such that $R(p_1, \ldots, p_s, P) = 0$.
  We plug the representation of $P$ by~\eqref{eq:P_expansion} into this equality and consider the result as polynomial in $u^{(h + h_0)}$.
  Since neither of $p_1, \ldots, p_s, Q, \frac{\partial p}{\partial u^{(h)}}$ involves $u^{(h + h_0)}$, every coefficient of this polynomial must vanish.
  The leading coefficient only involves $p_1, \ldots, p_s, \frac{\partial p}{\partial u^{(h)}}$ and thus yields an algebraic dependence of $\frac{\partial p}{\partial u^{(h)}}$ over $F$.
  Since $\frac{\partial p}{\partial u^{(h)}}$ has fewer monomials than $p$, all its monomials are algebraic over $F$. 
  By subtracting corresponding monomials from $p$, we obtain a polynomial with fewer monomials, so the induction hypothesis implies that the remaining coefficients of $p$ are also algebraic over $F$.
  The statement is proved.
  
  In order to prove the lemma, we apply the statement above to $p_1, \ldots, p_s$ and deduce that each element of $\mathcal{C}$ is algebraic over $F$.
  Therefore
  \[
    s \geqslant\operatorname{trdeg}_{k(u^{(\infty)})} k(\mathcal{C}) = \operatorname{trdeg}_{k} k(\mathcal{C}).
  \]
  On the other hand, $p_1, \ldots, p_s$ are algebraic over $k(u^{(\infty)}, \mathcal{C})$, so $\operatorname{trdeg}_{k(u^{(\infty)})} k(\mathcal{C}) \geqslant s$.
\end{proof}

\begin{corollary}\label{cor:alg_dep}
  Let $\Sigma$ be a system of the form~\eqref{eq:ODEmodel}.
  Let $p_1, \ldots, p_s \in k(\mathbf{x}, u)[(u')^{(\infty)}]$ be algebraically independent over $k(u^{(\infty)})$ such that, for every $1 \leqslant i \leqslant s$, $\mathcal{L}_{\Sigma}(p_i)$ is algebraic over $k(u^{(\infty)}, p_1, \ldots, p_s)$.
  
  Let $\mathcal{C}$ be the set of the coefficients of $p_1, \ldots, p_s$ considered as polynomials in $u', u'', \ldots$.
  Then $\operatorname{trdeg}_{k(u)} k(u, \mathcal{C}) = s$.
\end{corollary}

\begin{proof}
    We will modify $\Sigma$ by considering $u$ as a state variable $x_0$ and $u'$ an input $v$ and adding an equation $x_0' = v$.
    Applying Lemma~\ref{lem:alg_dep} to $p_1, \ldots, p_s, x_0$, we show that $\operatorname{trdeg}_k k(u, \mathcal{C}) = s + 1$, so $\operatorname{trdeg}_{k(u)} k(u, \mathcal{C}) = s$.
\end{proof}

\begin{proof}[Proof of Theorem~\ref{thm:sussmann}]
  Assume that $P$ is realizable by system $\Sigma$ as in~\eqref{eq:ODEmodel} of dimension $n$.
  For every $i \geqslant 0$, we define $p_i := \mathcal{L}_\Sigma^i(g)$ and observe that $p_h$ is algebraic over $k(u^{(\infty)}, p_0, \ldots, p_{h - 1})$.
  We denote the map $\mathbb{A}_{\mathbb{U}}^n \dashrightarrow \mathcal{H}_P$ given by Lemma~\ref{lem:linear_system} by $\bm{\gamma}$ (note that~$\bm{\gamma} = (p_0,\ldots,p_h)$) and the matrix of the system~\eqref{eq:linear_system_main} by $J_p$ (since it is the Jacobian of $p_i$'s with respect to $\mathbf{x}$).
  
  By renumbering $x_1, \ldots, x_n$ if necessary, we will assume that the minor of $J_{p}$ formed by the first $h$ columns is nonsingular.
  Let $\varphi\colon \mathbb{A}_{\mathbb{U}}^h \to \mathbb{A}_{\mathbb{U}}^n$ be a map such that $\varphi^*(x_i) = x_i$ for every $1 \leqslant i \leqslant h$ and $\varphi^*(x_i) \in \mathbb{Z}$ for $i > h$, and these integers are chosen in such a way so that $\varphi^*(J_p)$ is well-defined and the minor formed by the first $h$ columns of $\varphi^*(J_p)$ is nonsingular.
  We set $\widetilde{\bm{\gamma}} := \bm{\gamma} \circ \varphi$.
  By construction of $\varphi$, the Jacobian $\widetilde{J}_p$ of $\varphi^*(p_0), \ldots, \varphi^*(p_{h - 1})$ has rank $h$.
  Thus, by \cite[Theorem~2.2]{Ehrenborg1993}\footnote{The theorem is stated over $\mathbb{C}$ but the proof works for every field of zero characteristic}, they are algebraically independent over~$\mathbb{U}$.
  Then $\widetilde{\bm{\gamma}}$ is dominant.
  
  \textbf{Rational realizations.}
  For every $1 \leqslant i < h$, let $m_{i, 1}, \ldots, m_{i, N_i}$ be the list of monomials of $p_i$ as a polynomial in $u',  u'', \ldots$, and let $c_{i, 1}, \ldots, c_{i, N_i}$ be the corresponding list of coefficients.
  We denote $\mathcal{C} := (c_{i, j} \mid 0 \leqslant i < h, \; 1 \leqslant j \leqslant N_i)$.
  Then Corollary~\ref{cor:alg_dep} implies $\operatorname{trdeg}_{k(u)}k(u, \mathcal{C}) = h$.
  We will factor matrix $J_p$ as follows.
  Set $N := N_0 + \ldots + N_{h - 1}$.
  We define an $h \times N$-matrix $U$ such the $i$-th row is of the form
  \[
  (\underbrace{0, \ldots, 0}_{N_0 + \ldots + N_{i - 2}\text{ zeroes}}, m_{i - 1, 1}, \ldots, m_{i - 1, N_{i - 1}}, \underbrace{0, \ldots, 0}_{N_{i} + \ldots + N_{h - 1}\text{ zeroes}}).
  \]
  Then $J_p = U\cdot J_{\mathcal{C}}$, where $J_{\mathcal{C}}$ is the Jacobian of $\mathcal{C}$.
  For every $0 \leqslant i < h$, the monomials of $p_{i + 1} - D_u(p_i)$ are among $m_{i, 1}, \ldots, m_{i, N_i}$.
  Therefore, there exists $b \in (k(\mathbf{x}, u))^N$ such that the left-hand side of~\eqref{eq:linear_system_main} can be written as $U\cdot b$.
  Since the elemens of each row of $U$ are distinct monomials in $u', u'', \ldots$, the right kernel of $U$ over $k(\mathbf{x}, u)$ is zero.
  Hence, for every $v \in (k(\mathbf{x}, u))^n$, we have
  \begin{equation}\label{eq:cancel_u}
    U \cdot b = U \cdot J_{\mathcal{C}}\cdot v \iff b = J_{\mathcal{C}}\cdot v
  \end{equation}
  The system~\eqref{eq:linear_system_main} has a solution over $k(\mathbf{x}, u)$ due to the realizability of $P$.
  Then, by~\eqref{eq:cancel_u}, the same holds for 
  \begin{equation}\label{eq:expanded_system}
    b = J_{\mathcal{C}} \cdot \begin{pmatrix}
      Z_1 & \ldots & Z_n
    \end{pmatrix}^T,
  \end{equation}
  Since $\operatorname{trdeg}_{k(u)} k(\mathcal{C}, u) = h$, we have $\operatorname{rank} J_\mathcal{C} = h$ by \cite[Theorem~2.2]{Ehrenborg1993}.
  Then, due to our choise of ordering on $x$'s, $b$ belongs to the $k(\mathbf{x}, u)$-span of the first $h$ columns of $J_{\mathcal{C}}$.
  Then the same is true over $k(x_1, \ldots, x_h, u)$ for $\varphi^\ast(b)$ and $\widetilde{J}_{\mathcal{C}}$ which is formed by the first $h$ columns of $\varphi^\ast(J_{\mathcal{C}})$.
  Therefore, the system 
  \[
    U\cdot \varphi^\ast(b) = U \cdot \widetilde{J}_{\mathcal{C}} \cdot \begin{pmatrix}
      Y_1 & \ldots & Y_h
    \end{pmatrix}^T = \widetilde{J}_p \cdot \begin{pmatrix}
      Y_1 & \ldots & Y_h
    \end{pmatrix}^T
  \]
  has a solution in $k(x_1, \ldots, x_h, u)$.
  Thus, by Lemma~\ref{lem:linear_system}, $P$ has a realization of dimension $h$.

  \textbf{Input-affine rational realizations.}
  Consider the linear system~\eqref{eq:linear_system_main} provided by Lemma~\ref{lem:linear_system}.
  We will decompose each $Z_i$ as $Z_{i, 0} + Z_{i, 1}u$ and rewrite the system~\eqref{eq:linear_system_main} as a linear system with matrix $\begin{pmatrix}
    J_p & uJ_p
  \end{pmatrix}$ in variables $Z_{1, 0}, \ldots, Z_{n, 0}, Z_{1, 1}, \ldots, Z_{n, 1}$.
  Each solution of the new system in $k(\mathbf{x})$ gives rise to a solution of~\eqref{eq:linear_system_main} in $k(\mathbf{x}) + k(\mathbf{x})u$ and vice versa.
  For this new system we repeat the construction used to obtain~\eqref{eq:expanded_system} but considering monomials in $u^{(\infty)}$, not in $(u')^{(\infty)}$.
  We will obtain the following linear system over $k(\mathbf{x})$:
    \begin{equation}\label{eq:expanded_system2}
      b = J \cdot \begin{pmatrix}
        Z_{1, 0} & \ldots & Z_{n, 0} & Z_{1, 1} & \ldots & Z_{n, 1}
      \end{pmatrix}^T,
  \end{equation}
  Let $\mathcal{C}$ be again the list of coefficients of $p_0, \ldots, p_{h - 1}$.
  Since $up_0, \ldots, up_{h - 1}$ have the same coefficients but in front of different monomials, we can write $J = \begin{pmatrix}
    J_0 & J_1
  \end{pmatrix}$, where the rows of each of $J_0$ and $J_1$ are the rows of the Jacobian $J_{\mathcal{C}}$ and zero rows.
  Since $\operatorname{trdeg}_k k(\mathcal{C}) = h$, the dimension of the column space of each of $J_0$ and $J_1$ is equal to $h$.
  As in the rational case, this implies that the following system has a solution over $k(x_1, \ldots, x_h)$
  \[
    \varphi^*(b) = \begin{pmatrix}\widetilde{J}_0 & \widetilde{J}_1 \end{pmatrix}\begin{pmatrix}
        Y_{1, 0} & \ldots & Y_{h, 0} & Y_{1, 1} & \ldots & Y_{h, 1}
      \end{pmatrix}^T,
  \]
  where $\widetilde{J}_0$ and $\widetilde{J}_1$ are formed by the first $h$ columns of $\varphi^*(J_0)$ and $\varphi^*(J_1)$, respectively.
  This solution yields a solution of the corresponding system~\eqref{eq:linear_system_main} for $\widetilde{\bm{\gamma}}$ in $k(x_1, \ldots, x_h) + k(x_1, \ldots, x_h)u$.
  Hence Lemma~\ref{lem:linear_system} implies that $P$ has a realization of order $h$.

\end{proof}



\subsection{Realizability criteria for $\ord_u P \leqslant 1$}

\begin{proposition}\label{prop:order_zero}
  Let $P(y, u) \in k[y^{(\infty)}, u^{(\infty)}]$ be with $\ord_y P = h$ and $\ord_u P = 0$.
  Then there exists a rational (resp., input-affine rational) realization of $P$ if and only if there exists a dominant map $\bm{\gamma} \colon \mathbb{A}_{\mathbb{U}}^h \dashrightarrow \mathcal{H}_P$ such that $\gamma_0, \ldots, \gamma_{h - 1} \in k(\mathbf{x})$ and $\gamma_h \in k(\mathbf{x}, u)$ (resp., $\gamma_h \in k(\mathbf{x}) + k(\mathbf{x}) u$), where $\mathbf{x}$ are the coordinates in $\mathbb{A}^h_{\mathbb{U}}$.
\end{proposition}

\begin{proof}
  Assume that such a parametrization $\bm{\gamma}$ exists.
  We will show that it satisfies the conditions of Lemma~\ref{lem:linear_system}.
  The dominance of $\bm{\gamma}$ and the dependence of $P$ on $y^{(h)}$ imply that $\gamma_0, \ldots, \gamma_{h - 1}$ are algebraically independent over $k(u^{(\infty)})$.
  Thus, by \cite[Theorem~2.2]{Ehrenborg1993}, the matrix of the system~\eqref{eq:linear_system_main} is nonsingular.
  For the case of rational realization, we observe that~\eqref{eq:linear_system_main} is defined over $k(\mathbf{x}, u)$, so the unique solution will also be in this field.
  For the case of rational input-affine realization, we observe that the matrix of the system if defined over $k(\mathbf{x})$ and the entries of the left-hand side are in $k(\mathbf{x}) + k(\mathbf{x})u$.
  Therefore, Kramer's rule implies that the unique solution will also be in $k(\mathbf{x}) + k(\mathbf{x})u$.
  
  Now assume that $P$ is realizable. 
  By Theorem~\ref{thm:sussmann}, it is realizable by a system~$\Sigma$ as in~\eqref{eq:ODEmodel} of dimension $h$.
  Then $P$ vanishes at $g, \mathcal{L}_{\Sigma}(g), \ldots, \mathcal{L}_{\Sigma}^h(g)$.
  For every $R \in k(\mathbf{x}, u^{(\infty)})$ depending on $u$, we have
  $\ord_u \mathcal{L}_{\Sigma} R = \ord_u R + 1$.
  Let $h_0$ be the smallest integer~$0 \leqslant i \leqslant h$ such that $\mathcal{L}_{\Sigma}^{h_0}(g)$ involves $u$.
  Then $\ord_u \mathcal{L}_{\Sigma}^h(g) = h - h_0$.
  Since $u^{(h - h_0)}$ does not occur in $g, \mathcal{L}_{\Sigma}(g), \ldots, \mathcal{L}_{\Sigma}^{h - 1}(g)$, it must occur in $P$.
  Thus, $h = h_0$, so $(g, \mathcal{L}_\Sigma(g), \ldots, \mathcal{L}_\Sigma^h(g))$ yields a desired parametrization.
\end{proof}

\begin{proposition}\label{prop:spec_param_fo}
  Let $P(y, u) \in k[y^{(\infty)}, u^{(\infty)}]$ be with $\ord_y P = h$ and $\ord_u P = 1$.
  Then there exists a rational realization of $P$ if and only if there exists a dominant rational map $\bm{\gamma} \colon \mathbb{A}_{\mathbb{U}}^h \dashrightarrow \mathcal{H}_P$ (with the coordinates in $\mathbb{A}^h_\mathbb{U}$ denoted by $\mathbf{x}$) such that 
  \begin{itemize}
      \item $\gamma_0, \ldots, \gamma_{h - 2} \in k(\mathbf{x})$ and $\gamma_{h - 1} \in k(\mathbf{x}, u)$;
      \item $\gamma_{h} \in k(\mathbf{x}, u, u')$ and  $\frac{\partial \gamma_h}{\partial u'} = \frac{\partial \gamma_{h - 1}}{\partial u}$.
  \end{itemize}
\end{proposition}

\begin{proof}
  Assume that such a parametrization $\bm{\gamma}$ exists.
  We will show that it satisfies the conditions of Lemma~\ref{lem:linear_system}.
  As in the proof of Proposition~\ref{prop:order_zero}, the matrix of~\eqref{eq:linear_system_main} is nonsingular.
  Since $\gamma_h - D_u(\gamma_{h - 1})$ does not involve $u'$, the system~\eqref{eq:linear_system_main} is defined over $k(\mathbf{x}, u)$, so its unique solution belongs to $k(\mathbf{x}, u)$.
  
  Now assume that $P$ is realizable. 
  By Theorem~\ref{thm:sussmann}, it is realizable by a system~$\Sigma$ as in~\eqref{eq:ODEmodel} of dimension $h$.
  Similarly to the proof of Proposition~\ref{prop:order_zero}, one can show that $g, \mathcal{L}_{\Sigma}(g), \ldots, \mathcal{L}_{\Sigma}^{h - 2}(g) \in k(\mathbf{x})$, $\mathcal{L}_{\Sigma}^{h - 1}(g) \in k(\mathbf{x}, u)$, and $\mathcal{L}_{\Sigma}^h(g) \in k(\mathbf{x}, u, u')$.
  Furthermore, the definition of the Lie derivative implies that $\frac{\partial\mathcal{L}_{\Sigma}^h(g)}{\partial u'} = \frac{\partial \mathcal{L}^{h - 1}(g)}{\partial u}$.
  Therefore, $(g, \mathcal{L}_\Sigma(g), \ldots, \mathcal{L}_\Sigma^h(g))$ yields a desired parametrization.
\end{proof}

\section{Order zero in inputs}\label{sec:zero-ord}

In this section, we will consider the case of the input-output equation being
\begin{equation}\label{eq:P_zero_order}
  P(y, y', \ldots, y^{(h)}, u) = 0,
\end{equation}
that is, of zero order with respect to the input.
Proposition~\ref{prop:order_zero} reduces the realization problem for~\eqref{eq:P_zero_order} to finding a rational parametrization of the corresponding surface $\mathcal{H}_P$ of a special form.
Thus, it is sufficient to provide an algorithm for finding such a parametrization.
We show that finding such special parametrization over $k(u)$ can be reduced to finding rational parametrizations of several $h$-dimensional hypersurfaces over $k$ (Algorithm~\ref{alg:special_param}).
In particular, this yields complete algorithms for the cases $h = 1, 2$.
However, the resulting  procedure can be used in practice for $h > 2$ as well, see Examples~\ref{ex:SIR} and~\ref{ex:predator_prey1}.

\begin{algorithm}
\caption{Computing parametrization over $k / k(u)$}\label{alg:special_param}
\begin{description}[itemsep=0pt]
\item[Input ] \begin{itemize}
    \item irreducible polynomial $P \in k(u)[z_0, \ldots, z_h]$ depending nontrivially on $z_h$;
    \item a black-box algorithm for computing rational parametrizations of hypersurfaces over $k$ of dimension $h$.
\end{itemize}
\item[Output ] a parametrization $\bm{\gamma} = (\gamma_0, \ldots, \gamma_h)$ of $P(z_0, \ldots, z_h) = 0$ such that $\gamma_0, \ldots, \gamma_{h - 1}$ are defined over $k$ and $\gamma_h$ is defined over $k(u)$ if such a parametrization exists, and \texttt{NO} otherwise.
\end{description}

\begin{enumerate}[label = \textbf{(S\arabic*)}, leftmargin=*, align=left, labelsep=2pt, itemsep=0pt]
    \item Clear the denominators and further assume that $P \in k[z_0, \ldots, z_h, u]$ is irreducible.
    \item Write $P$ as a polynomial in $z_h$:
    \[
      P = A_d z_h^d + A_{d - 1} z_h^{d - 1} + \ldots + A_0,
    \]
    where $A_0, A_1,\ldots, A_d \in k[z_0, \ldots, z_{h - 1}, u]$.
    \item\label{step:shift} Apply shift $u \to u + c$ for $c \in k$ to ensure that $u \not\mid A_d$.
    \item\label{step:degrees} Let $d_0 := \deg_u A_0$, $d_1 := \deg_u A_d$, and introduce new variables $a_0, \ldots, a_{d_0}$ and $b_1, \ldots, b_{d_1}$.
    \item\label{step:Q} Substitute $z_h$ in $P$ with $\frac{a_0 + a_1 u + \ldots + a_{d_0}u^{d_0}}{1 + b_1 u + \ldots + b_{d_1}u^{d_1}}$. Denote the numerator of the resulting rational function by $Q$. 
    \item Denote the coefficients of $Q$  w.r.t. $u$ by $F_1, \ldots, F_N$.
    \item\label{step:rur} Compute a rational univariate representation~\cite{Rouillier1999} of the zero set of $F_1 = \ldots = F_N = 0$ considered as a polynomial system over $K := k(z_0, \ldots, z_{h - 1})$:
    \[
      q(w) = 0, a_0 = g_0(w), \ldots, b_{d_1} = g_{d_0 + d_1}(w),
    \]
    where $q \in K[T]$ and $g_0, \ldots, g_{d_0 + d_1} \in K(T)$.
    \item For each irreducible (over $K$) factor $r$ of $q$:
    \begin{enumerate}
        \item\label{step:comp_param} Check if there exists a rational parametrization $\bm{\alpha}$ of $r = 0$ in the space with coordinates $(z_0, \ldots, z_{h - 1}, w)$.
        \item If it exists, compute $a_0^\ast, \ldots, a_{d_0}^\ast, b_1^\ast, \ldots, b_{d_1}^\ast$ by evaluating $g_0, \ldots, g_{d_0 + d_1}$ at $\bm{\alpha}$ and \textbf{return}
        \[
        \bm{\gamma} := \left( \alpha_0, \ldots, \alpha_{h - 1}, \frac{a_0^\ast + a_1^\ast u + \ldots + a_{d_0}^\ast u^{d_0}}{1 + b_1^\ast u + \ldots + b_{d_1}^\ast u^{d_1}}\right).
        \]
    \end{enumerate}
    \item \textbf{Return} \texttt{NO}.
\end{enumerate}
\end{algorithm}

\begin{proposition}\label{prop:correctness_one}
  Algorithm~\ref{alg:special_param} is correct.
\end{proposition}

\begin{proof}
  First we will prove that the system $F_1 = \ldots = F_N = 0$ over $K$ (see~\ref{step:rur}) has dimension zero thus justifying that it is possible to compute rational univariate representation at step~\ref{step:rur}.
  Let $Z := \frac{a_0 + a_1 u + \ldots + a_{d_0}u^{d_0}}{1 + b_1 u + \ldots + b_{d_1}u^{d_1}}$.
  Since $\widetilde{P} := P(z_0, \ldots, z_{h - 1}, Z, u)$ and $Q$ from step~\ref{step:Q} differ by a factor $\frac{1}{(1 + b_1 u + \ldots + b_{d_1}u^{d_1})^M}$, the coefficients of $\widetilde{P}$ as a formal power series in $u$ belong to the ideal generated by $F_1, \ldots, F_N$.
  Thus, it is sufficient to prove that the coefficients of $\widetilde{P}$ generate a zero-dimensional ideal over $K$.
  We will prove this by showing that their Jacobian has full rank.
  We compute the partial derivatives of $\widetilde{P}$ w. r. to $a_i$'s and $b_j$'s:
  \begin{align}
  \begin{split}\label{eq:partial_diff}
      \frac{\partial\widetilde{P}}{\partial a_i}  & = \frac{\partial P}{\partial Z} \cdot \frac{u^{i}}{1 + b_1 u + \ldots + b_{d_1}u^{d_1}},\\
      \frac{\partial \widetilde{P}}{\partial b_j} & = \frac{\partial P}{\partial Z} \cdot \frac{-u^{j}}{(1 + b_1 u + \ldots + b_{d_1}u^{d_1})^2}.
  \end{split}
  \end{align}
  The coefficients of the above derivatives as power series in $u$ are the entries of the Jacobian, so it is sufficient to prove linear independence of these rational functions over $K(a_0, \ldots, a_{d_0}, b_1, \ldots, b_{d_1})$.
  Multiplying all the functions~\eqref{eq:partial_diff} by $\frac{(1 + b_1 u + \ldots + b_{d_1}u^{d_1})^2}{\partial P/\partial Z}$, we reduce the problem to verifying linear independence of the following polynomials:
  \begin{align*}
    &1 + b_1 u + \ldots + b_{d_1}u^{d_1}, \ldots, (1 + b_1 u + \ldots + b_{d_1}u^{d_1}) u^{d_0},\\
    &u,\; u^2,\; \ldots,\; u^{d_1},
  \end{align*}
  which is straightforward.
  
  Now we will prove that, if the algorithm returns $\bm{\gamma}$, such $\bm{\gamma}$ is a parametrization satisfying the output specification.
  Since $\gamma_i = \alpha_i$ for $0 \leqslant i < h$, we deduce that  $\gamma_0, \ldots, \gamma_{h - 1}$ are defined over $k$.
  The formula for $\gamma_h$ implies that it is defined over $k(u)$.
  In order to show that $\bm{\gamma}$ is a parametrization of $P = 0$, we observe the fact that $\bm{\alpha}$ is a parametrization of $q$ implies that $F_1, \ldots, F_N$ vanish under the substitution:
  \[
    z_i \to \alpha_i,\quad a_i \to a_i^\ast,\quad b_i \to b_i^\ast.
  \]
  Therefore, $P$ must vanish after substituting each $z_i$ with $\gamma_i$.
  Finally, the dominance of the map $\bm{\gamma}$ follows from the algebraic independence of $\gamma_0 = \alpha_0, \ldots, \gamma_{h - 1} = \alpha_{h - 1}$ due to the definition of $\bm{\alpha}$.
  
  Finally, we will show that if a parametrization $\bm{\gamma}$ satisfying the output specification of the algorithm exists, then such a parametrization will be found by the algorithm.
  We denote the variables used in the parametrization by $x_1, \ldots, x_h$ and write $\gamma_{h} = \frac{C}{D}$, where $C, D \in k(x_1, \ldots, x_h)[u]$ and $\gcd(C, D) = 1$.
  If we consider $P(\gamma_0, \ldots, \gamma_{h - 1}, z_h, u)$ to be a polynomial in a variable $z_h$ over a ring $k(x_1, \ldots, x_h)[u]$, $\gamma_{h}$ will be a rational function root of this polynomial.
  Therefore, its numerator $C$ and denominator $D$ divide the constant and leading terms of the polynomial, respectively.
  Therefore, $\deg_u C \leqslant d_0$ and $\deg_u D \leqslant d_1$.
  Furthermore, thanks to the shift at step~\ref{step:shift}, $D$ must not be divisible by $u$, so its constant term will be non-zero, and can be normalized to be one.
  After such normalization, we see that $\gamma_0, \ldots, \gamma_{h - 1}$ together with the coefficients of $C$ (as $a_i$'s) and $D$ (as $b_i$'s) yield a solution of the system $F_1 = \ldots = F_N = 0$.
  Due to the $k$-algebraic independence of $\gamma_0, \ldots, \gamma_{h - 1}$, this solution, via an isomorphism $K \cong k(\gamma_0, \ldots, \gamma_{h - 1})$, yields a solution of the corresponding zero-dimensional system over $K$ and thus must annihilate $q$ (see step~\ref{step:rur}).
  Therefore, $\gamma_0, \ldots, \gamma_{h - 1}$ together with the linear combination of the coefficients of $C$ and $D$ corresponding to the linear combination of $a_i$'s and $b_i$'s used to form $w$ annihilate $q$ and form a rational parametrization of one of its irreducible factors.
\end{proof}

\begin{remark}[On the input-affine case]
  Thanks to Lemma~\ref{prop:order_zero}, Algorithm~\ref{alg:special_param} can be used to find input-affine rational parametrizations as well.
  The only difference that the ansatz for $z_h$ constructed in~\ref{step:Q} should be taken simply $a_0 + a_1 u$.
\end{remark}


\section{First-order DAE}\label{sec:first-ord}

The goal of this section is to propose algorithms for finding rational and input-affine rational realizations of first-order (both in~$y$ and in~$u$) DAEs.

\subsection{Reminder on rational solutions for DAEs}

In this section we will recall and slightly refine the results from~\cite{Vo2018} about strong rational general solution of first-order DAEs. We start with giving the definition of rational curves. 

\begin{definition}[Rational parametrizations/curves]
Let~$V$ be an irreducible curve in~$\mathbb{A}^n$. 
A rational map $\mathcal{P}: \mathbb{A}^1 \dashrightarrow V$ defined by the set of rational functions~$\mathcal{P}(t) = (\chi_1(t),\ldots,\chi_n(t))$ is called a \emph{rational parametrization} of~$V$ if the following conditions are satisfied:
\begin{enumerate}
    \item $(\chi_1(t_0),\ldots,\chi_n(t_0))\in V$ for all (except for maybe a finite number of values) $t_0 \in k$.
    
    \item For all (except for maybe a finite number of values) points $p \in V$ there exists~$t_0 \in k$ such that~$p=(\chi_1(t_0),\ldots,\chi_n(t_0))$.
\end{enumerate}

A curve~$V$ is called \emph{rational} if it has a rational parametrization.

Any parametrization~$\mathcal{P}(t)$ induces a homomorphism~$\mathcal{P}^*: k(V) \rightarrow k(t)$. If~$\mathcal{P}^*$ is an isomorphism,~$\mathcal{P}(t)$ is called \emph{proper}.
\end{definition}

We will use the following refinement of Algorithm~\ref{alg:special_param}.

\begin{lemma}\label{lem:proper}
  Assume that $h = 1$ and the parametrization computed in step~\ref{step:comp_param} of Algorithm~\ref{alg:special_param} is proper.
  Then the parametrization returned by Algorithm~\ref{alg:special_param} is proper as well.
\end{lemma}

\begin{proof}
  Assume that the produced parametrization $\bm{\gamma} = (\gamma_0, \gamma_1)$ is not proper.
  This means that $k(u, \gamma_0, \gamma_1) \subsetneq k(x, u)$.
  Therefore, by~\cite[Theorem 9.29, p. 117]{milneFT} there exists an automorphism $\sigma$ of $\overline{k(x, u)}/k$ such that $\sigma|_{k(u, \gamma_0, \gamma_1)} = \operatorname{id}$ and $\sigma(x) \neq x$.
  Since $\gamma_0(\sigma(x)) = \gamma_0(x) \in k(x)$, we deduce that $\sigma(x) \in \overline{k(x)}$.
  Therefore, $\sigma$ can be restricted to $\overline{k(x)}$.
  Since $\sigma$ fixes $\gamma_1$  and $u$, and $u$ is transcendental over $\overline{k(x)}$, $\sigma$ fixes the coefficients of $\gamma_1$, that is, $a_0^\ast, \ldots, a_{d_0}^\ast, b_1^\ast, \ldots, b_{d_1}^\ast$.
  Since $\alpha_1$ is a $\mathbb{Q}$-linear combination of these, it is also fixed by $\sigma$.
  This contradicts the properness of $\bm{\alpha}$.
\end{proof}

The algorithm deciding the existence of a realization of a first order input-output equation by a rational dynamical system that we present in the next subsection is based on the notion of a \emph{strong rational general solution (SRGS)}. SRGS is a solution of an algebraic ODE (AODE) of the form \begin{equation} \label{AODE_SRGS}
    P (u,y, \operatorname{d}y / \operatorname{d}u ) = 0,
\end{equation} 
that depends rationally on a transcendental constant (for a precise definition, see~\cite[Definition 3.3]{Vo2018}).
Here~$P$ is an irreducible polynomial with coefficients in an algebraically closed field~$k$. Note that any SRGS of (\ref{AODE_SRGS}) provides a parametrization of the curve in~$\mathbb{A}^2$ defined by (\ref{AODE_SRGS}) over $\overline{k(u)}$ or any larger algebraically closed field (e.g., $\mathbb{U}$) with the transcendental constant arising in the solution being the parameter.

\begin{proposition}[{cf.~\cite[Theorem 3.7(iii)]{Johann}}] \label{proper_SRGS}
If the equation~\eqref{AODE_SRGS} has an SRGS, it also has an SRGS defining a proper parametrization of~$\mathcal{H}_P$. 
\end{proposition}

\begin{proof}
The algorithm for computing an SRGS of~\eqref{AODE_SRGS} presented in \cite{Vo2018} consists of finding an optimal proper parametrization of the corresponding rational curve and plugging a rational general solution of the \emph{associated ODE} (see~\cite[Definition 5.1]{Vo2018}) into this parametrization.

\cite[Theorem 5.4]{Vo2018} states that if an SRGS of~\eqref{AODE_SRGS} exists, an associated ODE is either a Riccati equation or a linear first-order equation. 
It is known (see, for instance, \cite[Sections A1.2 and A1.3]{Murphy}) that such equations have general solutions that are linear rational functions with respect to the constant of integration. 
In particular, if such an equation has a rational general solution, it necessarily has a rational general solution that is a linear rational function with respect to the constant of integration. 
By definition, such a solution is an SRGS.
This SRGS defines a linear rational substitution of the parameter on~$\mathcal{H}_P$ with the new parameter being the constant of integration. 
Since a linear rational transformation of the parameter transforms a proper parametrization into a proper one~\cite[Lemma 4.17]{SendraWinkler}, the claim is proved.
\end{proof}

\begin{proposition} \label{prop:constants}
Let~$y_1(u, c_1)$ and~$y_2(u,c_2)$ be two SRGS of~\eqref{AODE_SRGS} with~$y_1(u, c_1)$ providing a proper parametrization of  the corresponding curve. Then there exists $\varphi \in k(t)$ such that~$y_1(u, \varphi(c_2)) = y_2(u, c_2)$.
\end{proposition}

\begin{proof}
Since both~$(y_1, \partial y_1 / \partial u)$ and~$(y_2, \partial y_2 / \partial u)$ provide parametrizations of the same curve and the parametrization corresponding to~$y_1$ is proper, by \cite[Lemma 4.17]{SendraWinkler}, there exists~$\varphi\in\overline{k(u)}(t)$ such that:
\[
  y_1(u, \varphi(c_2)) = y_2(u, c_2) \; \text{and}\; \left( \partial y_1(u, c_1) / \partial u \right)\big\vert_{c_1 = \varphi(c_2)} = \partial y_2(u, c_2) / \partial u.
\]
Differentiating the former with respect to~$u$, we obtain
\[
  \left(\dfrac{\partial y_1(u, c_1)}{\partial u}\right)\big\vert_{c_1 = \varphi(c_2)} + \left(\dfrac{\partial y_1(u, c_1)}{\partial c_1}\right)\big\vert_{c_1 = \varphi(c_2)}\dfrac{\operatorname{d}\varphi(c_2)}{\operatorname{d}u} = \dfrac{\partial y_2(u, c_2)}{\partial u}.
\]
By combining the equations above, we obtain \[
\left(\dfrac{\partial y_1(u, c_1)}{\partial c_1}\right)\big\vert_{c_1 = \varphi(c_2)}\dfrac{\operatorname{d}\varphi(c_2)}{\operatorname{d}u} = 0,
\]
so~$\dfrac{\operatorname{d}\varphi(c_2)}{\operatorname{d}u}=0$, which means that~$\varphi \in k(t)$. 
\end{proof}


\subsection{Algorithm for rational realizations}\label{sec:sgrs_algo}

\begin{algorithm} 
\caption{First order realizations}\label{alg:first_order}
\begin{description}[itemsep=0pt]
\item[Input ] 
    Irreducible polynomial $P \in k[y,y',u,u']$ depending nontrivially on $y'$ and~$u'$;
\item[Output ] A system of the form (\ref{eq:ODEmodel}) such that~$P=0$ is the input-output equation for the system or \texttt{NO} if there is no such system.
\end{description}

\begin{enumerate}[label = \textbf{(S\arabic*)}, leftmargin=*, align=left, labelsep=2pt, itemsep=0pt]
    \item Let $Q := P(y, au' + b, u, u') \in k[y, a, b, u, u']$.
    Let $c_0 \in k[a, y, u]$ and $c_1 \in k[b, y, u]$ be the leading and the constant coefficients of $Q$ as a polynomial in $u'$.
    
    \item \label{step:parametrizations} For each irreducible factor $h_0(a, y, u)$ of $c_0$, do:
    \begin{enumerate}
        \item \label{step:SRGS_comp} Compute $y_0(u, c)$, an SRGS of $h_0\left(\dfrac{\operatorname{d}y}{\operatorname{d}u}, y, u\right) = 0$ defining a proper parametrization of the corresponding curve.
        If there is no solution, go to the next factor.
        \item For each irreducible factor $h_1(b, y, u)$ of $c_1$, do
        \begin{enumerate}
            \item Consider the numerator $N(b, c) \in k(u)[b, c]$ of $h_1(b, y_0(u, c), u)$.
            \item\label{step:inner_parametrization} Use Algorithm~\ref{alg:special_param} to find a proper parametrization (see Lemma~\ref{lem:proper}) $(b(x), c(x))$ of $N(b, c) = 0$ such that $b(x) \in k(u, x)$ and $c \in k(x)$.
            \item Set 
            \[ 
            g(x, u) := y_0(u, c(x))\;\text{ and }\; f(x, u) := b(x) / \frac{\partial g(x, u)}{\partial x}.
            \]
            \item\label{step:check} If the input-output equation of  $\Sigma$ being
            \[
              x' = f(x, u),\quad y = g(x, u)
            \]
            is equal to $P$, \textbf{return} $\Sigma$.
        \end{enumerate}
    \end{enumerate}
    
    \item\label{step:fail} \textbf{Return} \texttt{NO}.
\end{enumerate}
\end{algorithm}

\begin{proposition}\label{prop:alg2}
Algorithm \ref{alg:first_order} is correct. 
\end{proposition}

\begin{proof}
First we note that, thanks to the step~\ref{step:check}, if an ODE system is returned, it satisfies the specification of the algorithm.
Therefore, it is sufficient to prove that if the polynomial $P$ is realizable, the algorithm will return its realization.

Assume that $P$ is realizable.
By Theorem~\ref{thm:sussmann}, it is realizable by a one-dimensional system, we will denote it by $\Sigma_0$:
\[
  x' = f_0(x, u), \quad y = g_0(x, u).
\]
We also compute $\mathcal{L}_{\Sigma_0} (g_0) := \frac{\partial g_0}{\partial u}(x, u) u' + b_0(x, u)$.
Since 
\[
P(g_0, \mathcal{L}_{\Sigma_0} (g_0), u, u') = 0, 
\]
we have $c_0(\frac{\partial g_0}{\partial u}(x, u), g_0, u) = 0$, so $g_0(x, u)$ is an SRGS of one of the irreducible factors of $c_0$, we will denote this factor by $h_0$.
Consider the SGRS $y_0(u, c)$ computed by the algorithm for $h_0$ at step~\ref{step:SRGS_comp}.
By Proposition~\ref{prop:constants}, there exists a rational function $\varphi(x) \in k(x)$ such that $g_0(x, u) = y_0(u, \varphi(x))$.

The vanishing of $P(g_0, \mathcal{L}_{\Sigma_0} (g_0), u, u')$ implies that $c_1(b_0, g_0, u) = 0$, so $(b_0, g_0)$ annihilates one of the irreducible factors of $c_1$, say $h_1$.
Let $(b(x, u), c(x))$ be the parametrization computed at the step~\ref{step:inner_parametrization}. 
Then we have
\[
   h_1(b_0(x, u), y_0(u, \varphi(x)), u) = h_1(b(x, u), y_0(u, c(x)), u) = 0.
\]
The properness of the parametrization $(b(x, u), c(x))$ (see Lemma~\ref{lem:proper}) implies that there exists $s(x) \in k(x)$ such that $b_0(x, u) = b(s(x), u)$ and $\varphi(x) = c(s(x))$.
The system $\Sigma$ produced by the algorithm will have the following $f$ and $g$:
\[
  g(x, u) = y_0(u, c(x)), \quad f(x, u) = \frac{b(x, u)}{c'(x) z(c(x), u)},
\]
where $z := \frac{\partial y_0(x, c)}{\partial c}$.
The Lie derivative of $g$ will be
\begin{multline*}
  \mathcal{L}_\Sigma(g) = \frac{\partial y_0(u, c(x))}{\partial u} u' + c'(x) z(c(x), u) \frac{b(x, u)}{c'(x) z(c(x), u)} =\\= \frac{\partial y_0(u, c(x))}{\partial u} u' + b(x, u).  
\end{multline*}
Therefore, the pair $g_0, \mathcal{L}_{\Sigma_0}(g_0)$ can be obtained from $g, \mathcal{L}_{\Sigma}(g)$ by a substitution $x \mapsto s(x)$.
Therefore $P$ vanishes at $g, \mathcal{L}_{\Sigma}(g)$, so it is the input-output equation for $\Sigma$.
\end{proof}

\subsection{Algorithm for input-affine realizations}\label{sec:input_affine_algo}

\begin{algorithm}
\caption{First order input-affine realizations}\label{alg:first_order_ia}
\begin{description}[itemsep=0pt]
\item[Input ] 
    Irreducible polynomial $P \in k[y,y',u,u']$ depending nontrivially on $y'$ and~$u'$;
\item[Output ] An input-affine system of the form (\ref{eq:ODEmodel}) such that~$P=0$ is the input-output equation for this system or \texttt{NO} if there is no such system.
\end{description}

\begin{enumerate}[label = \textbf{(S\arabic*)}, leftmargin=*, align=left, labelsep=2pt, itemsep=0pt]
    \item\label{step:plug} Consider $\widetilde{P} = P(a_1u + a_0, b_2u^2 + b_1u + b_0 + a_1 u', u, u')$. 
    
    \item Consider the leading coefficient of $\widetilde{P}$ w.r.t. $u'$ and, in this coefficient, the leading coefficient~$q\in k[a_0,a_1]$ w.r.t.~$u$.
    
    \item\label{step:iaparams} For each irreducible factor~$q_0(a_0,a_1)$ of~$q$ do:
    \begin{enumerate}
        \item Find a proper rational parametrization~$a_0(s), a_1(s)$ of the curve~$q_0(a_0,a_1) = 0$. 
        If the curve is not rational, go to the next irreducible factor. 
        
        \item Plug~$a_0(s), a_1(s)$ into~$\widetilde{P}$ to obtain a polynomial in~$k(s, b_0, b_1, b_2)[u,u']$, and clear the denominators to obtain~$P_0 \in  k[s, b_0, b_1, b_2, u, u']$.
        
        \item Let $V$ be the variety defined by the coefficients of $P_0$ with respect to $u, u'$ in the space with coordinates $s, b_0, b_1, b_2$ over the field $k$.
        Then $\dim V \leqslant 1$ (see proof of Proposition~\ref{prop:alg3}).
        
        \item\label{step:component_loop} For each irreducible one-dimensional component~$V_0$ of $V$:
        \begin{enumerate}
            \item\label{step:param_s} Find a proper rational parametrization~$s(x), b_0(x), b_1(x), b_2(x)$ of~$V_0$. 
            If~$V_0$ is not rational, move to the next component. 
            
            \item Let~$a_i(x) := a_i(s(x))$ for $i = 0, 1$. 
            
            \item Let~$c_1(x):=b_2(x) / \dfrac{\operatorname{d}a_1}{\operatorname{d}x}$ and~$c_0(x):=b_0(x) / \dfrac{\operatorname{d}a_0}{\operatorname{d}x}$.
            
            \item\label{step:return_ia} \textbf{Return } $\begin{cases}
             x' = c_1(x)u+c_0(x),\\
             y = a_1(x)u+a_0(x).
             \end{cases}$
            
        \end{enumerate}
    \end{enumerate}
    
   \item\label{step:endgame} \textbf{Return} \texttt{NO}

\end{enumerate}
\end{algorithm}

\begin{proposition}\label{prop:alg3}
Algorithm \ref{alg:first_order_ia} is correct.
\end{proposition}

\begin{proof}
First we prove that the system returned in step \ref{step:return_ia} (denote it by~$\Sigma$) satisfies the specification of the algorithm.
This is ensured by Proposition~\ref{prop:spec_param_fo} combined with
\begin{multline*}
   P_0(s(x), b_0(x), b_1(x), b_2(x), u, u') = 0 \implies \\  P(y(x,u), \mathcal{L}_\Sigma(y(x,u)),u,u') = 0.
\end{multline*}

Assume~$P=0$ is realizable by an input-affine system.
By Theorem~\ref{thm:sussmann} it is realizable by a one-dimensional input-affine system, we will denote it by~$\Sigma_0$:
\[
x' = c_{0, 1}(x)u + c_{0, 0}(x),\quad y = g_0(x, u) = a_{0, 1}(x)u + a_{0, 0}(x).
\]
Since~$P$ vanishes at~$g_0, \mathcal{L}_{\Sigma_0}(g_0)$, $q$ vanishes at~$a_{0, 0}(x), a_{0, 1}(x)$.
Thus, at least one irreducible factor of~$q$, say~$q_0$, vanishes at~$a_{0, 0}(x), a_{0, 1}(x)$. 
Let~$a_0(s), a_1(s)$ be the parametrization of~$q_0=0$ obtained on step~\ref{step:param_s}. 
By \cite[Lemma 4.17]{SendraWinkler} there exists~$r(x) \in k(x)$ such that $a_{0, i}(x) = a_i(r(x))$ for $i = 0, 1$.

Analogously to the proof of Proposition~\ref{prop:correctness_one}, one can prove that coefficients of~$P_0$ define a zero-dimensional variety over~$k(s)$, i.e. a variety of dimension at most one over~$k$. 

We write $\mathcal{L}_{\Sigma_0}(g_0) = a_{0, 1}u' + b_{0, 2} u^2 + b_{0, 1} u + b_{0, 0}$.
Since~$P$ vanishes at~$g_0, \mathcal{L}_{\Sigma_0}(g_0)$, all the coefficients of~$P_0$ w.r.t. $u, u'$ vanish at~$r(x), b_0(x), b_1(x), b_2(x)$. 
This yields a parametrization of a one-dimensional component of $V$, say $V_0$.
Thus, once the algorithm will reach $V_0$ in step~\ref{step:component_loop}, it will find a realization.
\end{proof}


\section{Examples}\label{sec:examples}

Maple worksheets with all the details of the computations for the examples presented in this paper are available at~\cite{MapleNotebooks}.

\begin{example}[SIR model with input]\label{ex:SIR}
Consider a version of the standard SIR model from epidemiology which we have augmented with an input to the susceptible compartment (e.g., regulated travel of unvaccinated individuals):
\[
\begin{cases}
  S' = \Lambda - \mu S - \frac{\beta S I}{S + I + R} + u,\\
  I' = \frac{\beta S I}{S + I + R} - \mu I - \gamma I,\\
  R' = \gamma I - \mu R,
\end{cases}
\]
where $\Lambda, \mu, \beta, \gamma$ are scalar parameters and $u$ is the input.
We will assume that the output is $y = R$.
We will not give the full expressions arising in this computation due to their size, full details can be found in the Maple worksheet.

Computation using~\cite{structidjl} yields a differential equation in $y$ and $u$ of the respective orders $3$ and zero, so we will use Algorithm~\ref{alg:special_param}.
We have $d_0 = 1$ and $d_1 = 0$ thus the ansatz $y^{(3)} = a_1 u + a_0$ will be used.
This yields equations $F_1=0$ and $F_2=0$ such that $F_1$ is linear in $a_0$ and does not involve $a_1$ and $F_2$ is linear in $a_1$ and does not involve $a_0$.
Therefore, one can take $y = x_1, y' = x_2, y'' = x_3$ and extend this parametrization to $a_0$ and $a_1$ by solving linear equations.
This will yield a realization of the equation in $y$ and $u$ (different from the original ODE system).
\end{example}

For further examples we refer the reader to the appendix.

\begin{acks}
  GP was supported by the Paris Ile-de-France region (``XOR'' project) and NSF grants DMS-1853482, DMS-1760448, and
DMS-1853650.
  The authors are grateful to the referees for careful reading and numerous suggestions.
  The authors are grateful to Johann Mitteramskogler, Edurado Sontag, and Yuan Wang for helpful discussions.
\end{acks}

\bibliographystyle{ACM-Reference-Format}
\bibliography{bibliography}

\section*{Appendix: further examples}

\begin{example}\label{ex:SontagWang}
  The following example of a realization problem was considered in~\cite[Section 8]{SontagYuan}:
  \[
      uy'' = y^2 u^2 + y' u'.
  \]
  We will follow the same strategy as in Algorithm~\ref{alg:first_order}.
  Plugging~$y'' = au''+c$ and~$y' = au'+b$ into this equation yields
  $$-u^2y^2+auu''-au'^2-bu'+cu = 0$$
  The coefficient at~$u''$ must be zero, thus,~$a=0$,~$y' = b$ and~$y'' = pu' + q$, where~$b,p,q$ do not involve~$u'$ but can, in principle, involve~$u$.
  Plugging these expressions for~$y'$ and~$y''$ into the initial equation, we obtain
  $$(pu - b)u' - y^2u^2 + qu=0.$$
  Thus, we have
  $pu = b$ and $q-y^2u=0$.
  Note that~$p = \operatorname{d}b / \operatorname{d}u$, i.e. the first equation of this system is a DAE in~$b(u)$. Its SRGS is~$b = x_1u$, where~$x_1$ is a transcendental constant. The second equation of the system defines a rational curve over~$\overline{k(u)}$ with parametrization~$y=x_2, q=x_2^2u$. 
  Therefore
  $$
  y = x_2,\quad y' = x_1u,\quad y'' = x_1u' + ux_2^2.
  $$
  Using the approach from Lemma~\ref{lem:linear_system}, we find a realization
  $$
  x_1' = x_2^2,\quad  x_2' = x_1u,\quad y=x_2. 
  $$
\end{example}

\begin{example}[Predator-prey model]\label{ex:predator_prey1}
Consider the following predator-prey model with an input influencing the predator population:
\[
\begin{cases}
  x_1' = k_1 x_1 - k_2 x_1 x_2,\\
  x_2' = -k_3 x_2 + k_4 x_1 x_2 + k_5 u,\\
  y = x_1,
\end{cases}
\]
where $k_1, \ldots, k_5$ are scalar parameters.
The result of differential elimination (computed using the software~\cite{structidjl}) is
\begin{equation}\label{eq:pp_ord_zero}
y y'' - k_1 k_3 y^2 + k_1 k_4 y^3 + k_3 y y' + k_5 k_2 y^2 u - k_4 y^2 y' - (y')^2 = 0.
\end{equation}
We will now apply Algortihm~\ref{alg:special_param} with $z_0 = y, z_1 = y', z_2 = y''$.
We have $d_0 = 1$ and $d_1 = 0$ (see~\ref{step:degrees}), so we make an ansatz $y'' = a_1 u + a_0$ in~\ref{step:Q} and obtain:
\[
u (k_2 k_5 y^2 + a_1 y) + k_1 k_4 y^3 - k_1 k_3 y^2 - k_4 y^2 y' + k_3 y y' + a_0 y - (y')^2 = 0.
\]
We obtain the following equations for $a_0$ and $a_1$:
\[
\begin{cases}
  k_2 k_5 y^2 + a_1 y = 0 \implies a_1 = -k_2 k_5 y,\\
  k_1 k_4 y^3 - k_1 k_3 y^2 - k_4 y^2 y' + k_3 y y' + a_0 y - (y')^2 = 0.
\end{cases}
\]
One can take $w = a_0$ and $q(w) = 0$ to be the last equation.
The equation of this surface is linear in $a_0$, so the surface has a parametrization:
\[
  y = x_1, \; y' = x_2, \; a_0 = -k_1 k_4 x_1^2 + k_1 k_3 x_1 + k_4 x_1 x_2 - k_3 x_2 + \frac{x_2^2}{x_1}.
\]
We follow the proof of Lemma~\ref{lem:linear_system} and obtain a realization of~\eqref{eq:pp_ord_zero}:
\[
\begin{cases}
x_1' = x_2,\\
x_2' = -k_2 k_5 x_1 u - k_1 k_4 x_1^2 + k_1 k_3 x_1 + k_4 x_1 x_2 - k_3 x_2 + \frac{x_2^2}{x_1},\\
y = x_1
\end{cases}
\]
Note that, in the resulting ODE system, one can straighforwardly reduce the dimension of the parameter space by setting $k_6 := k_2k_5$ and thus providing a reparametrization of the original model.
\end{example}

\begin{example}[Predator-prey model, continued]\label{ex:predator_prey2}
  This example is version of Example~\ref{ex:predator_prey1} with a different choice of the output variable:
  \[
  \begin{cases}
    x_1' = k_1 x_1 - k_2 x_1 x_2,\\
    x_2' = -k_3 x_2 + k_4 x_1 x_2 + k_5 u,\\
    y = x_2.
   \end{cases}
   \]
   The input-output equation for this model is
   \begin{multline*}
     -k_5 k_1 y u + k_5 k_2 y^2 u + k_5 y u' - k_5 y' u + k_1 k_3 y^2 \\+ k_1 y y' - k_2 k_3 y^3 - k_2 y^2 y' - y y'' + (y')^2 = 0.  
   \end{multline*}
   We will follow the same strategy as Algorithm~\ref{alg:first_order}.
   Plugging~$y'' = au'' +c$ and~$y' = au'+b$ yields
   $$a^2u''^2+r(a,b,c,u,u') = 0$$
   for some rational function~$r$. Thus,~$a=0$,~$y' = b$ and~$y'' = pu' + q$, where~$b,p,q$ do not involve~$u'$ but can, in principle, involve~$u$. 
   The following equation is the result of plugging these expressions into the initial equation:
   \begin{multline*}
       (k_5 y - py)u' - k_2k_3y^3 + k_5k_2y^2u - bk_2y^2 + k_1k_3y^2 -\\- k_5k_1yu + bk_1y - bk_5u + b^2 - yq = 0.
   \end{multline*}
   The coefficient at~$u_1$ has to be zero, therefore,~$p=k_5$.
   
   Since the constant coefficient is linear in~$q$, it defines a unirational surface with parametrization~$y_0=v_1, b=v_2$ and
   $$q = -k_2k_3v_1^3+k_5k_2v_1u-v_2k_2v_1+k_1k_3v_1-k_5k_1u+v_2k_1-\dfrac{v_2k_5u}{v_1}+\dfrac{v_2^2}{v_1}.$$
   
   Since~$p = \dfrac{\operatorname{d}v_2}{\operatorname{d}u}$, one can conclude that~$v_2=k_5u+x_1$, where~$x_1$ is a transcendental constant, and~$v_2=x_2$. Thus,
   \begin{multline*}
   q(x_1,x_2,u)= (-k_2k_5x_2 + k_5k_1 + \dfrac{k_5^2u}{x_2} + \dfrac{k_5x_1}{x_2})u +\\+ k_2k_3x_2^2 + k_2k_5x_2u - k_1k_3x_2 - k_1k_5u + k_2x_1x_2 - \dfrac{k_5^2u^2}{x_2} - k_1x_1 - 2k_5u\dfrac{x_1}{x_2} - \dfrac{x_1^2}{x_2}
   \end{multline*}
   and we obtain the following realization of the original equation:
   $$
   \begin{cases}
   x_1' = q(x_1,x_2,u),\\
   x_2' = x_1+k_5u,\\
   y = x_2.
   \end{cases}
   $$
\end{example}

\end{document}